\documentclass[a4paper,11pt,reqno]{amsart}
\usepackage[margin=3cm]{geometry}
\usepackage[english]{babel}
\usepackage[T1]{fontenc}
\usepackage{latexsym,amsmath,amsfonts,amssymb,amsthm,hyperref,doi,tikz,float,arydshln,float}
\usetikzlibrary{arrows}
\newcommand{\arXiv}[2]{arXiv:\href{http://arxiv.org/abs/#1}{#1 #2}}

\hypersetup{%
pdfauthor={Tam\'{a}s~G\"{o}rbe and \'{A}d\'{a}m~Gyenge},%
pdftitle={Canonical spectral coordinates for the Calogero-Moser space associated with the cyclic quiver},%
pdfsubject={Preprint (\today)},%
pdfkeywords={Calogero-Moser, cyclic quiver, Darboux coordinates, canonical spectral coordinates, Sklyanin's formula}}

\theoremstyle{definition}

\newtheorem*{remark}{Remark}
\theoremstyle{plain}
\newtheorem{theorem}{Theorem}
\newtheorem{lemma}[theorem]{Lemma}
\newtheorem{proposition}[theorem]{Proposition}
\newtheorem{corollary}[theorem]{Corollary}

\def\C{\mathbb{C}}
\def\P{\mathbb{P}}
\def\R{\mathbb{R}}
\def\Z{\mathbb{Z}}
\def\1{\mathbf{1}}
\def\cC{\mathcal{C}}

\def\adj{\mathrm{adj}}
\def\diag{\mathrm{diag}}
\def\End{\mathrm{End}}
\def\GL{\mathrm{GL}}
\def\SL{\mathrm{SL}}
\def\tr{\mathrm{tr}}
\def\ri{\mathrm{i}}
\newcommand{\vect}[1]{\boldsymbol{#1}}
\def\blambda{\vect{\lambda}}

\begin{document}

\title[Sklyanin's formula for the CM space attached to the cyclic quiver]{Canonical spectral coordinates for the Calogero-Moser space associated with the cyclic quiver}

\author{Tam\'{a}s G\"{o}rbe}
\address{Tam\'{a}s G\"{o}rbe, School of Mathematics, University of Leeds, Leeds, LS2 9JT, UK}
\email{T.Gorbe@leeds.ac.uk}

\author{\'{A}d\'{a}m~Gyenge}
\address{\'{A}d\'{a}m~Gyenge, Department of Mathematics, The University of British Columbia, 1984 Mathematics Road, V6T 1Z2, Vancouver, BC Canada}
\email{agyenge@math.ubc.ca}

\maketitle

\begin{abstract}
Sklyanin's formula provides a set of canonical spectral coordinates on the standard Calogero-Moser space associated with the quiver consisting of a vertex and a loop. We generalize this result to Calogero-Moser spaces attached to cyclic quivers by constructing rational functions that relate spectral coordinates to conjugate variables. These canonical coordinates turn out to be well-defined on the corresponding simple singularity of type $A$, and the rational functions we construct define interpolating polynomials between them.
\end{abstract}

\section{Introduction}
\label{sec:intro}

The $n$-th Calogero-Moser space $\cC_n$ can be viewed as the completed phase space of the $n$-particle rational Calogero-Moser (CM) system \cite{Ca71,Mo75,Wi98}. This system describes $n$ interacting particles with positions $q=(q_1,\dots,q_n)$ and momenta $p=(p_1,\dots,p_n)$ evolving in time according to Hamilton's equations
\begin{equation}
\frac{dq_j}{dt}=\frac{\partial H}{\partial p_j},\quad \frac{dp_j}{dt}=-\frac{\partial H}{\partial q_j},\quad j=1,\dots,n
\label{eq:EoM}
\end{equation}
given by the Hamiltonian
\begin{equation}
H(p,q)=\sum_{j=1}^n\frac{p_j^2}{2}+\sum_{1\leq j<k\leq n}\frac{\gamma}{(q_j-q_k)^2}.
\label{eq:H}
\end{equation}
Here $\gamma$ is a parameter that controls the strength of particle interaction, which itself is defined via a pair-potential that is inversely proportional to the square of the difference of particle-positions. This system has many conserved quantities, i.e. functions $F$ such that $\{H,F\}=0$, that can be obtained as spectral invariants of a matrix-valued function of $q$ and $p$, that is the Lax matrix of the system. Moreover, the eigenvalues of the Lax matrix of the CM system form a complete set of Poisson commuting first integrals, hence the CM system is Liouville integrable \cite{Mo75}. This encourages the investigation of the spectrum of the Lax matrix. These eigenvalues provide partial parametrisation of the CM space on the dense open subset where the Lax matrix is diagonalisable. A natural question is to find a set of conjugate variables in order to obtain a full parametrisation compatible with the symplectic structure $\sum_{j=1}^ndp_j\wedge dq_j$. Sklyanin formulated a conjectural expression for conjugate variables in \cite{Sk09}. Utilizing the bi-Hamiltonian structure of the classical CM system, this conjecture was proved in \cite{FM17}. Another proof of Sklyanin's formula was given in \cite{Go16} using Hamiltonian reduction \cite{KKS78}.

Canonical spectral coordinates are central to the algebro-geometric approach to integrable systems \cite{Sk95}. In general, when the Lax matrix of a system depends on a spectral parameter $z$, canonical coordinates are given by the location of the poles of a (suitably normalized) eigenvector of the Lax matrix $L(z)$. Equivalently, the coordinates are given by the locations on the spectral curve $\det(\lambda\1-L(z))$ of the points corresponding to the zeros of a specific polynomial. However, this method cannot be applied directly to the rational CM-system, because all poles of the eigenvector are located above $z=\infty$, hence the $z$ coordinates of these poles do not provide conjugate variables to the eigenvalues of the specialized (spectral parameter independent) Lax matrix $L(\infty)$. A formula conjectured by Sklyanin \cite{Sk09} resolves exactly this problem. Instead of the coordinate $z$, some other function associated with the dynamical variables should be used to express the conjugate variables.

The classical CM space $\cC_n$ is also a particular example of a quiver variety \cite{Na94,Gi01}. Namely, it is associated with the quiver consisting of only one vertex with a loop attached to it. More general Calogero-Moser spaces associated with other quivers can be constructed in a similar manner. In this work, we investigate the CM space associated with the cyclic quiver as introduced in \cite{CS17}. For short, we will call it the \emph{equivariant Calogero-Moser space} since it can be thought of as the completed phase space of the equivariant $n$-particle rational Calogero-Moser system under the action of the cyclic group of order $m$. We will denote it by $\cC_{n}^m$. Our main observation is that, similarly to the non-equivariant case $\cC_n=\cC_n^1$, an explicit formula for the conjugate variables on $\cC_{n}^m$ can be given.

One can go even further by allowing the particles to have spin, i.e. internal degrees of freedom. The corresponding space will be denoted by $s \cC_n^m$, where we suppressed the dimension of the space of internal states. We extend our results to this case as well. Although the resulting formulas look the same as in the spinless case, the proofs differ at several points.

It was shown in \cite{CS17} that on the dense open subset where the specialized Lax matrix is diagonalisable, if $(\lambda_1,\dots,\lambda_n)$ are its the eigenvalues, then there is a certain set of variables $(\phi_1,\dots,\phi_n)$ which are conjugate to the eigenvalues. Our first main result is the following
\begin{theorem}
\label{thm:rel}
For each point of the Calogero-Moser space $\cC_{n}^m$ (resp., $s \cC_{n}^m$) there is a certain rational function $r(z)\in \C(z)$ such that on a dense open subset the relationship 
\begin{equation}
\phi_i=r(\lambda_i)\quad i=1,\dots,n
\end{equation}
between the sets of conjugate variables $(\phi_1,\dots,\phi_n)$ and $(\lambda_1,\dots,\lambda_n)$ holds.
\end{theorem}

Theorem \ref{thm:rel} shows that although the specialized Lax matrix of the equivariant CM system does not contain a spectral parameter, the conjugate variable pairs $(\lambda_i,\phi_i)$ are still lying on an ``interpolation curve'' defined by the equation $y=r(z)$. More precisely, the pair $(\lambda_i,\phi_i)$ is a well-defined point on the singular surface of type $A_{m-1}$, and the interpolation curve between the points $\{(\lambda_i,\phi_i)\}$ is a rational curve on this singular surface.

The datum which represents a point on $\cC_{n}^m$ or on $s \cC_{n}^m$ contains a \emph{framing}, which consists of two additional vectors $v$ and $w$. In the spin case the variables $(\lambda_1,\dots,\lambda_n)$ and $(\phi_1,\dots,\phi_n)$ together with the coordinates of the vectors $v$ and $w$ form a complete set of canonical coordinates, whereas in the spinless case the coordinates of $v$ and $w$ can always be gauged away and $(\lambda_1,\dots,\lambda_n,\phi_1,\dots,\phi_n)$ are enough for a complete local parametrisation.

It turns out that on $\cC_n$ (resp., $s \cC_{n}^m$) the set $(\phi_1,\dots,\phi_n)$ is not the only natural set of variables which is conjugate to $(\lambda_1,\dots,\lambda_n)$ \cite{Go16}.  The function $r(z)$ appearing in Theorem \ref{thm:rel} (including its special case for $m=1$) does not depend on the framing part of the datum whereas the conjectured formula in \cite{Sk09}, which gives another set of conjugate variables on $\cC_n$ in the $m=1$ case, does depend on the framing. Our second result is that an analogue of Sklyanin's formula from \cite{Sk09} is also valid in the equivariant case, and there is a second natural set of conjugate variables to $(\lambda_1,\dots,\lambda_n)$ which depends on the framing as well.
\begin{theorem}
\label{thm:mod}
For each point of the Calogero-Moser space $\cC_{n}^m$ (resp., $s \cC_{n}^m$) there is a certain rational function $s(z) \in \C(z)$, depending also on the framing part of the datum, such that on a dense open subset the variables $\theta_1,\dots,\theta_n$, defined as
\begin{equation}
\theta_i=s(\lambda_i),\quad i=1,\dots,n
\end{equation}
are conjugate to $\lambda_1,\dots,\lambda_n$, respectively.
\end{theorem}
It is known that the non-equivariant CM space $\cC_n$ is a deformation of the Hilbert scheme $\mathrm{Hilb}^n(\C^2)$ of $n$ points on $\C^2$. The framing vectors play an essential role in the stability condition of the GIT construction of $\mathrm{Hilb}^n(\C^2)$ as a quiver variety \cite{Na99}. Hence, it seems useful to keep track of the framing vectors (or their steadiness) during a degeneration of $\cC_n$ into $\mathrm{Hilb}^n(\C^2)$. The advantage of Theorem \ref{thm:mod} is that as opposed to $r(z)$ the functions $s(z)$ can measure such a steadiness.

Theorems \ref{thm:rel} and \ref{thm:mod} show that there are at least two natural sets of variables conjugate to the spectral variables $(\lambda_1,\dots,\lambda_n)$. Correspondingly, there are two natural interpolation curves on the singular surface of type $A_{m-1}$.

The structure of the paper is as follows. In Section \ref{sec:SoV} we recall the recipe of separation of variables and its relation to the spectral curve with a special emphasis on the rational CM system. In Section \ref{sec:quiver} we summarize the construction and the symplectic structure on the CM space associated with the cyclic quiver. In Section \ref{sec:CM-spaces} we prove Theorems \ref{thm:rel} and \ref{thm:mod} for the spinless case. In Section \ref{sec:spin-CM-spaces}, after introducing the equivariant CM space with spin, we give the proofs of Theorems \ref{thm:rel} and \ref{thm:mod} for this case. In Section \ref{sec:interpolation} we construct the interpolation curves on the singular surface of type $A_{m-1}$.

\subsection*{Acknowledgements.}
The authors would like to thank Jim Bryan and Bal\'{a}zs Szendr\H{o}i for helpful comments and discussions.

\bigskip\noindent
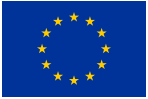
\parbox{.15\textwidth}{\begin{tikzpicture}[scale=.035]
\fill[fill={rgb,255:red,0;green,51;blue,153}] (-27,-18) rectangle (27,18);  
\pgfmathsetmacro\inr{tan(36)/cos(18)}
\foreach \i in {0,1,...,11} {
\begin{scope}[shift={(30*\i:12)}]
\fill[fill={rgb,255:red,255;green,204;blue,0}] (90:2)
\foreach \x in {0,1,...,4} { -- (90+72*\x:2) -- (126+72*\x:\inr) };
\end{scope}}
\end{tikzpicture}} \parbox{.85\textwidth}{This project has received funding from the European Union's Horizon 2020 research and innovation programme under the Marie Sk{\l}odowska-Curie grant agreement No 795471.}

\section{Separation of variables and the spectral curve of the rational Calogero-Moser system}
\label{sec:SoV}

We briefly review the method of separation of variables (SoV) following \cite{Sk95}. Consider a Liouville integrable system having $n$ degrees of freedom. This means a $2n$-dimensional symplectic manifold $(\mathcal{P},\omega)$ with $n$ independent smooth functions $H_1,\dots,H_n$ on it that commute with respect to the Poisson bracket $\{,\}$ induced by the symplectic form $\omega$, i.e. $\{H_j,H_k\}=\omega(X_{H_j},X_{H_k})=0$, $j,k=1,\dots,n$. A system of canonical coordinates $(p_j,q_j)$, $j=1,\dots,n$, i.e. local coordinates on the symplectic manifold satisfying
\begin{equation}
\{p_j,p_k\}=\{q_j,q_k\}=0, \quad \{p_j,q_k\}=\delta_{j,k}\quad j,k=1,\dots,n
\end{equation}
is called \emph{separated} if there exist $n$ relations of the form
\begin{equation}
\label{eq:sepeq}
\Phi_j(q_j,p_j,H_1,\dots,H_n)=0, \quad j=1,\dots,n.
\end{equation}
Such a system of variables induces an explicit decomposition of the Liouville tori into one-dimensional tori and makes several calculations about the system straightforward \cite{Sk95}.

Suppose that the system under consideration has a Lax representation. This means that the equations of motion \eqref{eq:EoM} can be written in the form
\begin{equation}
\dot{L}(z)=[L(z),M(z)]
\end{equation}
with some matrices $L(z)$ and $M(z)$ of size $n \times n$, whose elements are functions on the phase space and which depend on an additional parameter $z$ called spectral parameter. Then the functions $H_1,\dots,H_n$ can be expressed in terms of the coefficients $t_1(z),\dots,t_n(z)$ of the characteristic polynomial $W(\Lambda, z)$ of the matrix $L(z)$
\begin{equation}
W(\Lambda,z)=\det(\Lambda\1 - L(z))=\sum_{k=0}^n(-1)^kt_k(z)\Lambda^{n-k}.
\end{equation}
The characteristic equation
\begin{equation}
W(\Lambda,z)=0
\label{eq:chareq}
\end{equation}
defines the eigenvalue $\Lambda(z)$ of $L(z)$ as a function on the corresponding $n$-sheeted Riemannian surface of the parameter $z$. The Baker-Akhiezer function $\Omega(z)$ is defined as the eigenvector of $L(z)$ corresponding to the eigenvalue $\Lambda(z)$, i.e. we have
\begin{equation}
L(z)\Omega(z)=\Lambda(z)\Omega(z).
\label{eq:eigvect}
\end{equation}
After a suitable normalization, $\Omega(z)$ becomes a meromorphic function on the Riemannian surface \eqref{eq:chareq}. Sklyanin's formula hints that the coordinates $z_j$ of these poles play an important role. The formula is based on the observation that for many models the variables $z_j$ Poisson commute and, together with the corresponding eigenvalues $\Lambda_j = \Lambda(z_j )$ of $L(z_j)$, or some functions $p_j$ of $z_j$, provide a set of separated canonical variables for the Hamiltonians $H_1,\dots,H_n$.  One reason for this is that since $\Lambda_j=\Lambda(z_j )$ is an eigenvalue of $L(z_j)$, the pair $(\Lambda_j,z_j)$ lies on the spectral curve \eqref{eq:chareq}, i.e.
\begin{equation}
\label{eq:wljzj}
W(\Lambda_j,z_j)=0.
\end{equation}
If, furthermore, $z_j$ is a function of $p_j$, then \eqref{eq:wljzj} provides the equations \eqref{eq:sepeq} as well.

The (complexified) rational Calogero-Moser system is a completely integrable Hamiltonian system describing a collection of $n$ identical particles on the affine line $\C$. The phase space of the Calogero-Moser system is $T^{\ast}(\C^n \setminus\{\textrm{all diagonals}\})$, the configurations are $n$ distinct unlabelled points $q_j\in \C$ with momenta $p_j\in \C$. The Lax matrix of the system can be brought to the form \cite{Ca75,KBBT95} \cite[(53)]{Ta02}
\begin{equation}
L(z)=L+\ri gz^{-1}\vect{e}\vect{e}^\top,
\label{eq:spectrallax}
\end{equation}
where $\vect{e}\in\R^n$ is the vector given by
\begin{equation}
\vect{e}=(1,\dots,1)^\top,
\end{equation}
and the components of the $z$-independent matrix $L$ are
\begin{equation}
L_{j,k}=p_j\delta_{j,k}+\ri g(q_j-q_k)^{-1}(1-\delta_{j,k}),\quad j,k=1,\dots,n.
\label{}
\end{equation}
As it was observed in \cite[Section 5.2]{Ta02}, the matrix determinant lemma
\begin{equation}
\det(M+\vect{x}\vect{y}^\top)=\det(M)+\vect{y}^\top\adj(M)\vect{x}
\end{equation}
implies that the characteristic polynomial of $L(z)$ simplifies to
\begin{equation}
\det(\Lambda\1-L(z))=P_0(\Lambda)-\ri gz^{-1}P_1(\Lambda),
\end{equation}
where
\begin{equation}
P_0(\Lambda)=\det(\Lambda\1-L)
\end{equation}
is the characteristic polynomial of $L$, and
\begin{equation}
P_1(\Lambda)=\vect{e}^\top\adj(\Lambda\1-L)\vect{e}=\tr(\adj(\Lambda\1-L)\vect{e}\vect{e}^\top).
\end{equation}
(We note that $\adj(M)$ denotes the adjugate matrix of $M$.)
In particular, the characteristic equation of the spectral curve of the system takes the form of a graph of a rational function
\begin{equation}
z=\ri g\frac{P_1(\Lambda)}{P_0(\Lambda)}.
\end{equation}

It follows that if $|z|<\infty$, then the eigenvector equation \eqref{eq:eigvect} can always be solved, and the solution has a finite magnitude. This means that all poles of the Baker-Akhiezer function $\Omega(z)$ are at $z=\infty$. The eigenvalues of $L(\infty)$ are exactly the eigenvalues of the matrix $L$ due to \eqref{eq:spectrallax}. Let us denote these by $\lambda_1,\dots,\lambda_n$. They form one half of a set of conjugate variables. Since each $\lambda_j$ lies on the level set $z= \infty$, the function $z$ cannot be a conjugate variable to them on the moduli space of all solutions of the system.

A similar situation occurs for the open Toda chain, which was resolved in \cite[2.20b]{Sk13}. In that case one looks for another rational expression which provides the sought-after conjugate variables. For the classical CM system such an expression for conjugate variables was conjectured in \cite{Sk09}. The formula turns out to be again a rational function, depending on the eigenvalues $\lambda_1,\dots,\lambda_n$, the matrix $L$, and another matrix $X$, which, in a suitable basis, has the particle-positions $q_1,\dots,q_n$ along its diagonal. The formula was verified using two different approaches, first in \cite{FM17} and then in \cite{Go16}.

Our aim in the forthcoming sections is to adapt these methods to more general Calogero-Moser systems and the moduli spaces of their solutions. Formally, the resulting formulas for the Calogero-Moser space associated with the cyclic quiver look similar to the classical case \cite{FM17, Go16}. Hence, one may expect that the formulas may hold more generally to Calogero-Moser spaces associated with any ``nice'' quiver. For a detailed study of the geometry of moment maps for quiver representations, see \cite{Cr01}.

\section{Calogero-Moser space associated with the cyclic quiver}
\label{sec:quiver}

Let $m$ be a positive integer. In this section we introduce the  Calogero-Moser space associated with the affine Dynkin quiver $A^{(1)}_{m-1}$ shown in Figure \ref{fig:1} below.
\begin{figure}[h!]
\centering
\begin{tikzpicture}
\draw[thick,fill=black]
(0,0)circle(.1)node[yshift=-1em]{$1$}--
(1.5,0)circle(.1)node[yshift=-1em]{$2$}--
(3,0)node[fill=white,inner sep=1em]{$\dots$}--
(4.5,0)circle(.1)node[yshift=-1em]{$m-2$}--
(6,0)circle(.1)node[yshift=-1em]{$m-1$}--
(3,2.25)circle(.1)node[yshift=1em]{$0$}--(0,0);
\end{tikzpicture}
\caption{The cyclic quiver.}
\label{fig:1}
\end{figure}
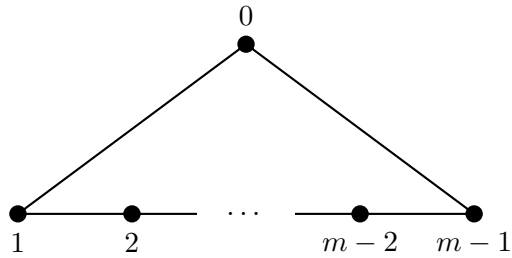

Starting from this quiver we first take the corresponding doubled quiver. This means that we replace each edge with a pair of edges with opposite orientation to each other. We also equip the quiver with a one-dimensional framing at the vertex 0, and construct the associated Calogero-Moser quiver variety. See Figure \ref{fig:2} for a particular example. The precise procedure of the construction is as follows.

\begin{figure}[ht!]
\centering
\begin{tikzpicture}[thick,rotate=90]
%radius of circles and number of vertices
\def\r{3} \def\N{7}
%drawing vertices and labelling nodes
\foreach \n [count=\ni from 0] in {1,...,\N}{
\draw({\r*cos(90+360/\N-\n*360/\N)},{\r*sin(90+360/\N-\n*360/\N)})circle(.3)node(V\ni){$V_{\ni}$};}
\draw(0,0)circle(.35)node(C){$V_\infty$};
%drawing edges
\foreach \n [count=\ni from 0] in {1,...,\N}{\pgfmathtruncatemacro{\nn}{mod(\ni+1,\N)}
\draw[->](V\ni) edge[bend left=30] node[fill=white,inner sep=1pt]{$X_\ni$} (V\nn)
(V\nn) edge[bend left=30] node[fill=white,inner sep=1pt]{$P_\ni$} (V\ni);}
\draw[->](V0) edge[bend left=15] node[fill=white,inner sep=1pt]{$w_0$}(C)
(C) edge[bend left=15] node[fill=white,inner sep=1pt]{$v_0$}(V0);
\end{tikzpicture}
\caption{The doubled cyclic quiver for $m=7$ with a special framing.}
\label{fig:2}
\end{figure}
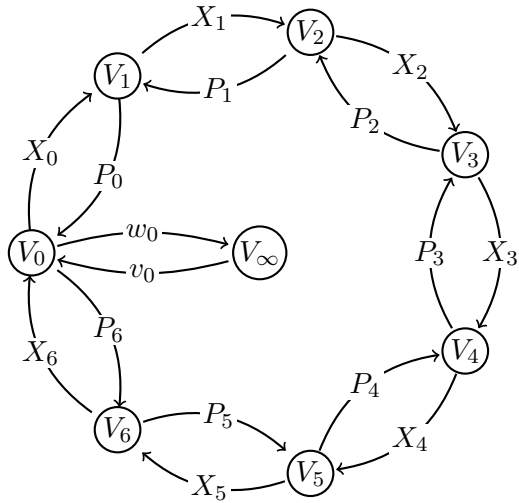

Fix a positive integer $n$ and let $V_0,V_1,\dots,V_{m-1}$ be vector spaces of dimension $n$ and $V_\infty$ be a one-dimensional vector space over the complex field $\C$. Let $\Z_m$ stand for the additive group $\Z/m\Z$ of integers modulo $m$, that is the cyclic group of order $m$. Let us consider the linear maps
\begin{equation}
X_i\colon V_i\to V_{i+1},\quad P_i\colon V_{i+1}\to V_i,\qquad i\in\Z_m
\label{}
\end{equation}
and by introducing a one-dimensional vector space $V_\infty$ over $\C$ we also define the linear maps
\begin{equation}
v_0\colon V_\infty\to V_0,\quad w_0\colon V_0\to V_\infty.
\label{v0-w0}
\end{equation}
Take the direct sum $V=V_0\oplus V_1\oplus\dots\oplus V_{m-1}$ and define the transformations $X,P\in\End(V)$ by
\begin{equation}
X(\vect{v}_0\oplus\vect{v}_1\oplus\dots\oplus\vect{v}_{m-1})=X_{m-1}\vect{v}_{m-1}\oplus X_0\vect{v}_0\oplus\dots\oplus X_{m-2}\vect{v}_{m-2}
\label{X}
\end{equation}
and
\begin{equation}
P(\vect{v}_0\oplus\vect{v}_1\oplus\dots\oplus\vect{v}_{m-1})=P_0\vect{v}_1\oplus P_1\vect{v}_2\oplus\dots\oplus P_{m-1}\vect{v}_0.
\label{P}
\end{equation}
Let $\1_V$ denote the identity map on $V$. The commutator $[X,P]\in\End(V)$ of  $X$ and $P$ can be expressed as
\begin{equation}
[X,P](\vect{v}_0\oplus\vect{v}_1\oplus\dots\oplus\vect{v}_{m-1})=\bigoplus_{i\in\Z_m}(X_{i-1}P_{i-1}-P_iX_i)\vect{v}_i.
\label{}
\end{equation}
Extend $v_0$, $w_0$ introduced in \eqref{v0-w0} to maps $v\colon V_\infty\to V$ and $w\colon V\to V_\infty$, respectively, by
\begin{equation}
v(z)=v_0(z)\oplus\vect{0}_{V_1}\oplus\dots\oplus\vect{0}_{V_{m-1}}
\quad\text{and}\quad
w(\vect{v}_0\oplus\vect{v}_1\oplus\dots\oplus\vect{v}_{m-1})=w_0(\vect{v}_0).
\label{v-w}
\end{equation}

An $m$-tuple $g=(g_0,g_1,\dots,g_{m-1})\in\C^m$ is called regular if
\begin{equation}
g_0+\dots+g_{m-1}\neq 0\quad\text{and}\quad k(g_0+\dots+g_{m-1})\neq g_h+\dots+g_{i-1}
\label{}
\end{equation}
for all $k\in\Z$ and $1\leq h<i\leq m-1$. We introduce $g\1_V\in\End(V)$ via
\begin{equation}
g\1_V(\vect{v}_0\oplus\vect{v}_1\oplus\dots\oplus\vect{v}_{m-1})=
g_0\vect{v}_0\oplus g_1\vect{v}_1\oplus\dots\oplus g_{m-1}\vect{v}_{m-1}.
\label{}
\end{equation}
Let $C_{n,g}^m$ stand for the space of quadruples $(X,Y,v,w)$ satisfying
\begin{equation}
[X,P]=g\1_V+vw.
\label{constraint}
\end{equation}
The group $\GL(V)\subset\End(V)$ of invertible linear transformations acts on $C_{n,g}^m$ by
\begin{equation}
M\cdot(X,P,v,w)=(MXM^{-1},MPM^{-1},Mv,wM^{-1}),\quad M\in\GL(V).
\label{}
\end{equation}
If $g$ is regular, then this action is free. The equivariant Calogero-Moser space $\cC_{n,g}^m$ for the cyclic group $\mathbb{Z}_m$ is defined as the space of orbits, i.e.
\begin{equation}
\cC_{n,g}^m\cong C_{n,g}^m/\GL(V).
\label{}
\end{equation}
In the rest of the paper we will suppress the dependence of $\cC_{n,g}^m$ on $g$, and simply write $\cC_{n}^m$.

In \cite{CS17} Chalykh and Silantyev introduced local coordinates on the subset $\cC_{n}^{m,X}\subset\cC_{n}^m$ consisting of orbits with invertible and diagonalisable maps $X_i$. Namely, they diagonalised each of the $X_i$ by choosing such an $M\in\GL(V)$ that $Q=MXM^{-1}$, when written in the standard basis, has an $m$-by-$m$ block matrix structure with blocks of size $n$. The non-zero blocks are at positions $(i+1,i)$ and are of the form
\begin{equation}
\vect{q}:=[Q]_{i+1,i}=\diag(q_1,\dots,q_n),\qquad i\in\Z_m.
\label{}
\end{equation}
By using the group action and the constraint \eqref{constraint} they showed that each point of $\cC_{n}^{m,X}$ can be represented by $(Q,L,Mv,wM^{-1})$ with $Q$ as displayed above and $L=MPM^{-1}$ having a similar block matrix structure with non-zero blocks at positions $(i,i+1)$ whose components are
\begin{equation}
(L_i)_{j,k}:=([L]_{i,i+1})_{j,k}=\big(p_j-c_iq_j^{-1}\big)\delta_{j,k}+|g|q_j^iq_k^{m-i-1}(q_j^m-q_k^m)^{-1}(1-\delta_{j,k}),
\label{}
\end{equation}
$i\in\Z_m$, $j,k\in\{1,\dots,n\}$, where $p_1,\dots,p_n$ are arbitrary and $c_i$, $|g|$ are constants, namely
\begin{equation}
c_i=\sum_{r=0}^ig_r-\sum_{s=0}^{m-1}\frac{m-s}{m}g_s
\quad\text{and}\quad
|g|=\sum_{s=0}^{m-1}g_s.
\label{c_i-|g|}
\end{equation}
The maps $Mv,wM^{-1}$ are expressed as column and row vectors, respectively, with $m$ blocks of size $n$ each. The only non-zero blocks are the first ones, i.e.
\begin{equation}
[Mv]_0=(1\quad 1\quad\dots\quad 1)^\top
\quad\text{and}\quad
[wM^{-1}]_0=-|g|(1\quad 1\quad\dots\quad 1).
\label{}
\end{equation}
It was also proved in \cite{CS17} that these local coordinates $(p_j/m,q_j)$ are canonical. That is, the symplectic structure on $\cC_{n}^{m,X}$, obtained from the standard symplectic form on $C_{n}^m$, can be written as
\begin{equation}
\omega=m\sum_{j=1}^n dp_j\wedge dq_j.
\label{}
\end{equation}
The Hamiltonians can be written as
\begin{equation}
H_k(p,q)=\frac{1}{mk}\tr(L^{mk}),\quad k=1,\dots,n,
\label{eq:H_k}
\end{equation}
and for $m=1$ we have $H_1(p,q)=H(p,q)$ \eqref{eq:H} with $\gamma=-g_0^2$.

The same procedure can be repeated by introducing local coordinates $(\lambda_j,\phi_j)$ on the subset $\cC_{n}^{m,P}\subset\cC_{n}^m$ consisting of orbits with diagonalisable maps $P_i$. We denote the corresponding objects by putting a tilde over them. Namely, we have an invertible transformation $\tilde{M}\in\GL(V)$ such that the matrix of $\tilde{L}=\tilde{M}P\tilde{M}^{-1}$ has diagonal blocks at positions $(i,i+1)$
\begin{equation}
\blambda:=[\tilde{L}]_{i,i+1}=\diag(\lambda_1,\dots,\lambda_n),\qquad i\in\Z_m,
\label{tilde-L}
\end{equation}
the matrix of $\tilde{Q}=\tilde{M}X\tilde{M}^{-1}$ has non-zero blocks at positions $(i+1,i)$
\begin{equation}
(\tilde{Q}_i)_{j,k}:=([\tilde{Q}]_{i+1,i})_{j,k}=\big(\phi_j+c_i\lambda_j^{-1}\big)\delta_{j,k}-|g|\lambda_j^{m-i-1}\lambda_k^i(\lambda_j^m-\lambda_k^m)^{-1}(1-\delta_{j,k}).
\label{}
\end{equation}
The maps $\tilde{v}=\tilde{M}v$ and $\tilde{w}=w\tilde{M}^{-1}$ can be written as vectors of size $mn$ with only the first $n$ entries being non-zero:
\begin{equation}
[\tilde{v}]_0=(1\quad 1\quad\dots\quad 1)^\top
\quad\text{and}\quad
[\tilde{w}]_0=|g|(1\quad 1\quad\dots\quad 1).
\label{}
\end{equation}
Finally, the symplectic structure on $\cC_{n}^{m,P}$ can be written as
\begin{equation}
\tilde{\omega}=m\sum_{j=1}^n d\lambda_j\wedge d\phi_j.
\label{eq:omegasumw}
\end{equation}
hence the Poisson bracket of two functions $f,g\in C^\infty(\cC_n^{m,P})$ is given by
\begin{equation}
\{f,g\}=m\sum_{j=1}^n\bigg(\frac{\partial f}{\partial\lambda_j}\frac{\partial g}{\partial\phi_j}
-\frac{\partial f}{\partial\phi_j}\frac{\partial g}{\partial\lambda_j}\bigg).
\label{eq:Poisson}
\end{equation}
The Hamiltonians $H_k$ \eqref{eq:H_k}, when expressed in terms of $(\lambda_j,\phi_k)$, take a much simpler form
\begin{equation}
H_k=\frac{1}{mk}\tr(\tilde{L}^{mk})=\frac{1}{k}(\lambda_1^{mk}+\dots+\lambda_n^{mk}),\quad k=1,\dots,n.
\end{equation}

\section{Canonical spectral coordinates in the spinless case}
\label{sec:CM-spaces}

Now we turn to the task of finding explicit formulas for variables conjugate to the eigenvalues $\lambda_1,\dots,\lambda_n$ of $P_i$, i.e. such functions $\theta_1,\dots,\theta_n$ in involution that
\begin{equation}
\{\lambda_j,\theta_k\}=\delta_{jk},\quad j,k=1,\dots,n.
\label{}
\end{equation}
It follows from \eqref{eq:Poisson} that the variables $\phi_1/m,\dots,\phi_n/m$ are
such functions. Proposition \ref{prop:1} below provides explicit formulas for $\phi_k$ in terms of $\lambda$. To formulate the statement we need the following functions on $C_n^m$ that depend on an extra variable $z$:
\begin{gather}
A(z)=\det(z\1_V-P),\label{eq:A}\\ 
C(z)=\tr(X\,\adj(z\1_V-P)vw),\label{eq:C}\\
D(z)=\tr(X\,\adj(z\1_V-P)).\label{eq:D}
\end{gather}
Notice that these functions, besides $z$, depend only on the \emph{class} of the quadruple $(X,P,v,w)$ under the $\GL(V)$-action (we have suppressed this dependence). Therefore $A,C,D$ descend to well-defined functions on $\cC_{n}^m$ for which we use the same notation. Here $X,P,v,w$ are given by \eqref{X}, \eqref{P}, \eqref{v-w} and $\adj$ denotes the adjugate map. We remark that $C(z)$ can also be written as
\begin{equation} C(z)=wX\adj(z\1_V-P)v.  \end{equation}

\begin{lemma}\label{lemma:1}
The characteristic polynomial $A(z)=\det(z\1_V-P)$ can be written in terms of $\lambda_1,\dots,\lambda_n$ as
\begin{equation}
A(z)=\prod_{j=1}^n(z^m-\lambda_j^m).
\label{}
\end{equation}
\end{lemma}

\begin{proof}[Proof \#1]
Notice that $A(z)$ is invariant under conjugation, i.e. constant along orbits of $\GL(V)$. Thus we can use $\tilde{L}=\tilde{M}P\tilde{M}^{-1}$ instead of $P$. Let us express $A(z)$ in the basis in which the matrix of $\tilde{L}$ is the one displayed in \eqref{tilde-L}. This means that $A(z)$ can be written as 
\begin{equation}
A(z)=\left\vert\begin{array}{c;{6pt/2pt}cccc}
z\1_n&-\blambda&0&\dots&0\\\hdashline[6pt/2pt]
0&z\1_n&-\blambda&\ddots&\vdots\\
\vdots&\ddots&\ddots&\ddots&\vdots\\
0&\ddots&\ddots&\ddots&-\blambda\\
-\blambda&0&\dots&0&z\1_n
\end{array}\right\vert_{m\times m},
\label{partition}
\end{equation}
where the index $m\times m$ indicates that the number of blocks in each row and column is $m$.
If we partition the matrix as indicated by the dashed lines and apply the block matrix determinant formula
\begin{equation}
\det\left[\begin{array}{c;{6pt/2pt}c}
\alpha&\beta\\\hdashline[6pt/2pt]
\gamma&\delta
\end{array}\right]=\det(\alpha)\det(\delta-\gamma\alpha^{-1}\beta),
\label{block-det-formula-1}
\end{equation}
(with the assumption that $z\neq 0$) then we get
\begin{equation}
A(z)=z^n\left\vert\begin{array}{c;{6pt/2pt}cccc}
z\1_n&-\blambda&0&\dots&0\\\hdashline[6pt/2pt]
0&z\1_n&-\blambda&\ddots&\vdots\\
\vdots&\ddots&\ddots&\ddots&\vdots\\
0&\ddots&\ddots&\ddots&-\blambda\\
-z^{-1}\blambda^2&0&\dots&0&z\1_n
\end{array}\right\vert_{(m-1)\times(m-1)}.
\label{}
\end{equation}
By iterating this process $(m-2)$ times we obtain
\begin{equation}
A(z)=z^{(m-2)n}\left\vert\begin{array}{c;{6pt/2pt}c}
z\1_n&-\blambda\\\hdashline[6pt/2pt]
-z^{-(m-2)}\blambda^{m-1}&z\1_n
\end{array}\right\vert_{2\times 2}.
\label{}
\end{equation}
Applying the determinant formula \eqref{block-det-formula-1} one more time yields
\begin{equation}
A(z)=z^{(m-1)n}\det(z\1_n-z^{-(m-1)}\blambda^m)=\det(z^m\1_n-\blambda^m)
=\prod_{j=1}^n(z^m-\lambda_j^m).
\label{}
\end{equation}
This concludes the proof.
\end{proof}

Let us give an alternative and more direct proof.

\begin{proof}[Proof \#2]
In this proof we partition the matrix the same way as in \eqref{partition}, but apply a different version of the block matrix determinant formula, namely
\begin{equation}
\det\left[\begin{array}{c;{6pt/2pt}c}
\alpha&\beta\\\hdashline[6pt/2pt]
\gamma&\delta
\end{array}\right]=\det(\delta)\det(\alpha-\beta\delta^{-1}\gamma).
\label{block-det-formula-2}
\end{equation}
This requires the calculation of the determinant and inverse of the bottom right block. Fortunately, this block is an upper triangular matrix of size $(m-1)n$. Its determinant is
\begin{equation}
\det(\delta)=z^{(m-1)n},
\label{}
\end{equation}
and (with the assumption that $z\neq 0$) its inverse exists and is, of course, also upper triangular. The $(h,i)$-th block of $\delta^{-1}$ is
\begin{equation}
[\delta^{-1}]_{h,i}=z^{h-i-1}\blambda^{i-h},\ \text{if}\ h\leq i,\qquad [\delta^{-1}]_{h,i}=\mathbf{0}_{n},\ \text{if}\ h>i.
\label{}
\end{equation}
The product $\beta\delta^{-1}\gamma$ is simply $\blambda^2[\delta^{-1}]_{1,m-1}$. Substituting everything into \eqref{block-det-formula-2} yields
\begin{equation}
A(z)=z^{(m-1)n}\det(z\1_n-z^{-(m-1)}\blambda^m)=\det(z^m\1_n-\blambda^m)=\prod_{j=1}^n(z^m-\lambda_j^m),
\label{}
\end{equation}
which concludes the proof.
\end{proof}

\begin{lemma}\label{lemma:2}
The inverse of $z\1_V-P$ can be written explicitly in terms of $\lambda_1,\dots,\lambda_n$ as an $m\times m$ block matrix with blocks of size $n$ of the form
\begin{equation}
[(z\1_{mn}-\tilde{L})^{-1}]_{h,i}=z^{m-(i-h+1)}(z^m-\blambda^m)^{-1}\blambda^{i-h},\quad h,i\in\Z_m,\label{z-P-inverse}
\end{equation}
where the exponents $m-(i-h+1)$ and $i-h$ are understood modulo $m$.

If one does not wish to use mod $m$ exponents one can write $[(z\1_{mn}-\tilde{L})^{-1}]$ as
\begin{align}
[(z\1_{mn}-\tilde{L})^{-1}]_{h,i}&=z^{h-i-1}(z^m-\blambda^m)^{-1}\blambda^{m-(h-i)},\quad  h,i=0,\dots,m-1,\ h>i,\label{z-P-inverse-1}\\
[(z\1_{mn}-\tilde{L})^{-1}]_{h,i}&=z^{m-(i-h+1)}(z^m-\blambda^m)^{-1}\blambda^{i-h},\quad h,i=0,\dots,m-1,\ h\leq i.\label{z-P-inverse-2}
\end{align}
\end{lemma}

\begin{proof}
A simple check confirms that the matrix defined by formulas \eqref{z-P-inverse-1}-\eqref{z-P-inverse-2} is such that $(z\1_{mn}-\tilde{L})(z\1_{mn}-\tilde{L})^{-1}=\1_{mn}$.
\end{proof}

We recall that the adjugate of an invertible linear transformation $M$ can be written as $\adj(M)=\det(M)M^{-1}$, hence assuming that $z\1_V-P$ is invertible we have the following
\begin{equation}
\adj(z\1_V-P)=A(z)(z\1_V-P)^{-1},
\label{eq:adj-det}
\end{equation}
where $A(z)=\det(z\1_V-P)$ as defined in \eqref{eq:A}.

The next statement gives Theorem \ref{thm:rel} for the spinless CM space $\cC_{n}^m$.
\begin{proposition}\label{prop:1}
For a point $[(X,P,v,w)] \in \mathcal{C}_n^m$ let 
\begin{equation}
r(z)=\frac{D(z)}{A'(z)} \in \C(z),
\end{equation}
where $A(z)$ and $D(z)$ are the functions defined in \eqref{eq:A} and \eqref{eq:D}, respectively. Then the variables $\phi_1,\dots,\phi_n$ can be expressed as
\begin{equation}
\phi_k=r(\lambda_k),\quad k=1,\dots,n.
\label{}
\end{equation}
\end{proposition}

\begin{proof}
Since $A(z)$ and $D(z)$ are both invariant under conjugation by elements of $\GL(V)$, using $\tilde{Q}=\tilde{M}X\tilde{M}^{-1}$ instead of $X$ and $\tilde{L}=\tilde{M}P\tilde{M}^{-1}$ instead of $P$ in these functions gives the same results. We already expressed $A(z)$ in terms of $\lambda_1,\dots,\lambda_n$ in Lemma \ref{lemma:1}, so let us consider $D(z)$ and calculate the diagonal blocks of $\tilde{Q}\,\adj(z\1_{mn}-\tilde{L})$. These blocks can be calculated by utilizing \eqref{eq:adj-det} and Lemma \ref{lemma:2}. Namely, we get
\begin{equation}
[\tilde{Q}\,\adj(z\1_{mn}-\tilde{L})]_{i,i}=\tilde{Q}_{i-1}A(z)z^{m-2}\blambda(z^m-\blambda^m)^{-1},\quad i\in\Z_m.
\label{diag-ii}
\end{equation}
The function $D(z)$ can be written in terms of $\lambda_1,\dots,\lambda_n$ as
\begin{equation}
D(z)=\tr(\tilde{Q}\,\adj(z\1_{mn}-\tilde{L}))=\sum_{i=0}^{m-1}\tr([\tilde{Q}\,\adj(z\1_{mn}-\tilde{L})]_{i,i}).
\label{eq:dzdef1}
\end{equation}
Plugging \eqref{diag-ii} into this formula gives
\begin{equation}
D(z)=\sum_{i=0}^{m-1}\tr(\tilde{Q}_i A(z)z^{m-2}\blambda(z^m-\blambda^m)^{-1})
=\sum_{i=0}^{m-1}\sum_{j=1}^n(\phi_j-c_i\lambda_j^{-1})\lambda_jz^{m-2}\prod_{\substack{\ell=1\\(\ell\neq j)}}^n(z^m-\lambda_\ell^m).
\label{D(z)}
\end{equation}
Substituting $z=\lambda_k$ causes all terms with $j\neq k$ to vanish leaving
\begin{equation}
D(\lambda_k)=\sum_{i=0}^{m-1}(\phi_k-c_i\lambda_k^{-1})\lambda_k^{m-1}\prod_{\substack{\ell=1\\(\ell\neq k)}}^n(\lambda_k^m-\lambda_\ell^m).
\label{D-lambda-k}
\end{equation}
Differentiating $A(z)$ with respect to $z$ yields
\begin{equation}
A'(z)=mz^{m-1}\sum_{j=1}^n\prod_{\substack{\ell=1\\(\ell\neq j)}}^n(z^m-\lambda_\ell^m),
\label{eq:azpdef1}
\end{equation}
which at $z=\lambda_k$ takes the following form
\begin{equation}
A'(\lambda_k)=m\lambda_k^{m-1}\prod_{\substack{\ell=1\\(\ell\neq k)}}^n(\lambda_k^m-\lambda_\ell^m).
\label{A'-lambda-k}
\end{equation}
Putting formulas \eqref{D-lambda-k} and \eqref{A'-lambda-k} together gives
\begin{equation}
\frac{D(\lambda_k)}{A'(\lambda_k)}=\frac{1}{m}\sum_{i=0}^{m-1}(\phi_k-c_i\lambda_k^{-1})=\phi_k+\frac{c_0+\dots+c_{m-1}}{m\lambda_k}.
\label{}
\end{equation}
By using \eqref{c_i-|g|} a simple calculation reveals that $c_0+\dots+c_{m-1}=0$ leaving us with
\begin{equation}
\frac{D(\lambda_k)}{A'(\lambda_k)}=\phi_k
\label{}
\end{equation}
and the proof is complete.
\end{proof}

Next, we will prove the  analogue of Sklyanin's formula \cite{Go16,Sk09} in the equivariant case, which provides another set of variables $\theta_1,\dots,\theta_n$ conjugate to $\lambda_1,\dots,\lambda_n$. The result gives Theorem \ref{thm:mod} for $\mathcal{C}_n^m$.
\begin{proposition}\label{prop:2}
For a point $[(X,P,v,w)] \in \mathcal{C}_n^m$ let us define the function
\begin{equation}
s(z)=\frac{C(z)}{|g|A'(z)} \in \C(z)
\end{equation}
with $A(z)$ and $C(z)$ defined in \eqref{eq:A} and \eqref{eq:C}, respectively, and use it to define the variables $\theta_1,\dots,\theta_n$ as
\begin{equation}
\theta_k=s(\lambda_k),\quad k=1,\dots,n.
\label{Sklyanin}
\end{equation}
Then $\theta_k$ can be written as
\begin{equation}
\theta_k=\frac{\phi_k}{m}+f_k(\lambda_1,\dots,\lambda_n),\quad k=1,\dots,n,
\label{15}
\end{equation}
with such $\lambda$-dependent functions $f_1,\dots,f_n$ that
\begin{equation}
\frac{\partial f_j}{\partial\lambda_k}=\frac{\partial f_k}{\partial\lambda_j},\quad
j,k=1,\dots,n.
\label{16}
\end{equation}
In particular, the variables $\theta_1,\dots,\theta_n$ given by \eqref{Sklyanin} are conjugate to $\lambda_1,\dots,\lambda_n$, i.e. we have $\{\theta_j,\theta_k\}=0$ and $\{\lambda_j,\theta_k\}=\delta_{j,k}$, $j,k=1,\dots,n$.
\end{proposition}

\begin{proof}
Let us start with $C(z)$. Using gauge invariance we replace the quadruple $(X,P,v,w)$ by $(\tilde Q,\tilde{L},\tilde{v},\tilde{w})$ just as we did before, to get
\begin{equation}
C(z)=\tr(\tilde{Q}\,\adj(z\1_{mn}-\tilde{L})\tilde{v}\tilde{w})=\tr([\tilde{Q}\adj(z\1_{mn}-\tilde{L})]_{0,0}[\tilde{v}\tilde{w}]_{0,0}).
\label{}
\end{equation}
Using \eqref{diag-ii} with $i=0$ yields
\begin{equation}
C(z)=\tr(\tilde{Q}_{m-1}A(z)z^{m-2}\blambda(z^m-\blambda^m)^{-1}[\tilde{v}\tilde{w}]_{0,0}).
\label{}
\end{equation}
Since $[\tilde{v}\tilde{w}]_{0,0}$ is the $n\times n$ matrix that has $|g|$ for all of its components we get
\begin{equation}
C(z)=|g|\sum_{j,t=1}^n(\tilde{Q}_{m-1})_{j,t}\lambda_tz^{m-2}\prod_{\substack{\ell=1\\(\ell\neq t)}}^n(z^m-\lambda_\ell^m).
\label{C(z)-1}
\end{equation}
Substituting $z=\lambda_k$ into this expression yields
\begin{equation}
C(\lambda_k)=|g|\sum_{j=1}^n(\tilde{Q}_{m-1})_{j,k}\lambda_k^{m-1}\prod_{\substack{\ell=1\\(\ell\neq k)}}^n(\lambda_k^m-\lambda_\ell^m).
\label{C-lambda-k}
\end{equation}
Putting formulas \eqref{A'-lambda-k} and \eqref{C-lambda-k} together gives
\begin{equation}
\theta_k=\frac{C(\lambda_k)}{|g|A'(\lambda_k)}=\frac{1}{m}\sum_{j=1}^n(\tilde{Q}_{m-1})_{j,k}=\frac{\phi_k}{m}+f_k(\lambda_1,\dots,\lambda_n)
\label{}
\end{equation}
with
\begin{equation}
f_k(\lambda_1,\dots,\lambda_n)=\frac{c_{m-1}}{m\lambda_k}-\frac{|g|}{m\lambda_k}\sum_{\substack{\ell=1\\(\ell\neq k)}}^n\frac{\lambda_k^m}{\lambda_k^m-\lambda_\ell^m}.
\label{eq:fkdef1}
\end{equation}
This implies that $\{\lambda_j,\theta_k\}=\delta_{j,k}$, $j,k=1,\dots,n$. The partial derivative of $f_k$ with respect to $\lambda_j$ ($j\neq k$) is
\begin{equation}
\frac{\partial f_k}{\partial\lambda_j}=-\frac{|g|(\lambda_j\lambda_k)^{m-1}}{(\lambda_k^m-\lambda_j^m)^2},
\label{}
\end{equation}
which is clearly invariant under exchanging $k$ and $j$, and therefore we have
\begin{equation}
\frac{\partial f_k}{\partial\lambda_j}=\frac{\partial f_j}{\partial\lambda_k}
\label{}
\end{equation}
entailing $\{\theta_j,\theta_k\}=0$ for all $j,k=1,\dots,n$. This concludes the proof.
\end{proof}

\section{Calogero-Moser spaces with spin variables and their canonical variables}
\label{sec:spin-CM-spaces}

In this section, we derive the analogues of the results obtained in the previous section to models containing spin variables. Let $d$ be a positive integer. The affine Dynkin quiver $A^{(1)}_{m-1}$ in this case is equipped with a framing of dimension $d$ at each of its vertices, or, equivalently, with one $d$ dimensional framing which is connected to every other node. This latter formulation will be more convenient for us. Accordingly, we redefine the maps $v,w$ \eqref{v-w} to be $v\colon\C^{dm}\to V$, $w\colon V\to\C^{dm}$ given by
\begin{equation}
v(z_0,\dots,z_{m-1})=v_0(z_0)\oplus v_1(z_1)\oplus\dots\oplus v_{m-1}(z_{m-1})
\label{v-spin}
\end{equation}
and
\begin{equation}
w(\vect{v}_0\oplus\vect{v}_1\oplus\dots\oplus\vect{v}_{m-1})
=(w_0(\vect{v}_0),w_1(\vect{v}_1),\dots,w_{m-1}(\vect{v}_{m-1})),
\label{w-spin}
\end{equation}
where $v_i\colon\C^d\to V_i$, $w_i\colon V_i\to\C^d$ ($i\in\Z_m$) are linear maps. Points in the (equivariant) spin Calogero-Moser space are represented by quadruples $(X,P,v,w)$ satisfying
\begin{equation}
[X,P]=g\1_V+vw.
\label{}
\end{equation}
The space itself is denoted as $s\cC_{n}^m$, where we have suppressed the dependence on the stability vector $g$ as well as on the dimension $d$ of the space of internal states.
Two dual models of this space, similar to the ones presented in Section \ref{sec:CM-spaces}, can be given. For details, see \cite[Subsection 5.4.]{CS17}.

\begin{figure}[ht!]
\centering
\begin{tikzpicture}[thick,rotate=90]
%radius of circles and number of vertices
\def\r{3} \def\N{7}
%drawing vertices and labelling nodes
\foreach \n [count=\ni from 0] in {1,...,\N}{
\draw({\r*cos(90+360/\N-\n*360/\N)},{\r*sin(90+360/\N-\n*360/\N)})circle(.3)node(V\ni){$V_{\ni}$};}
\draw(0,0)circle(.35)node[inner sep=.5em](C){$V_\infty$};
%drawing edges
\foreach \n [count=\ni from 0] in {1,...,\N}{\pgfmathtruncatemacro{\nn}{mod(\ni+1,\N)}
\draw[->](V\ni) edge[bend left=30] node[fill=white,inner sep=1pt]{$X_\ni$} (V\nn)
(V\nn) edge[bend left=30] node[fill=white,inner sep=1pt]{$P_\ni$} (V\ni);
\draw[->](V\ni) edge[bend left=15] node[fill=white,inner sep=1pt]{$w_\ni$}(C)
(C) edge[bend left=15] node[fill=white,inner sep=1pt]{$v_\ni$}(V\ni);}
\end{tikzpicture}
\caption{The doubled cyclic quiver for $m=7$ with the modified framing.}
\end{figure}
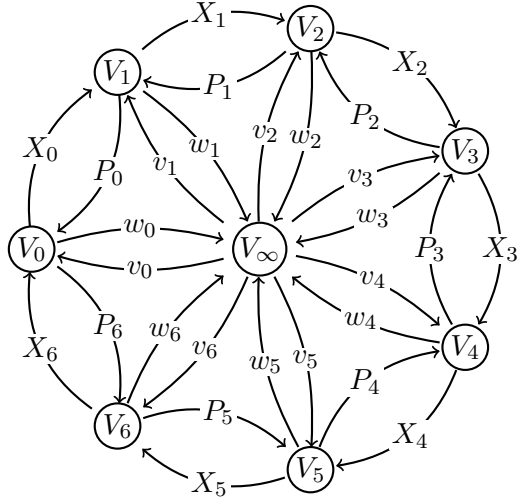

The objects we are most interested in are the ones corresponding to $\tilde{Q},\tilde{L}$. With a slight abuse of notation, we let $\tilde{Q},\tilde{L}$ denote the spin versions as well. They have the same block matrix structure as previously, but certain non-zero blocks are different. The map $\tilde{L}$ has the same matrix as before, so for non-zero blocks we have
\begin{equation}
\blambda=[\tilde{L}]_{i,i+1}=\diag(\lambda_1,\dots,\lambda_n),
\label{}
\end{equation}
while the non-zero blocks of the matrix of $\tilde{Q}$ are given by
\begin{equation}
(\tilde{Q}_i)_{j,j}=([\tilde{Q}]_{i+1,i})_{j,j}
=\phi_j+\lambda_j^{-1}\bigg(c_i+\sum_{r=0}^i[\tilde{v}_r \tilde{w}_r]_{j,j}
-\sum_{s=0}^{m-1}\frac{m-s}{m}[\tilde{v}_s\tilde{w}_s]_{j,j}\bigg)
\label{eq:tilde-Q-j,j}
\end{equation}
for $i\in\Z_m$, $j=1,\dots,n$ and
\begin{equation}
(\tilde{Q}_i)_{j,k}=([\tilde{Q}]_{i+1,i})_{j,k}=-\sum_{h=0}^{m-1}[\tilde{v}_{i-h} \tilde{w}_{i-h}]_{j,k}
\frac{\lambda_j^{m-h-1}\lambda_k^h}{\lambda_j^m-\lambda_k^m},
\label{eq:tilde-Q-j,k}
\end{equation}
for $i\in\Z_m$, $j,k=1,\dots,n$ $(j\neq k)$. (The index $i-h$ of $\tilde{v}$ and $\tilde{w}$ is understood modulo $m$.) The maps $\tilde{v},\tilde{w}$ have matrices that satisfy the equation
\begin{equation}
\sum_{i=0}^{m-1}[\tilde{v}_i \tilde{w}_i]_{j,j}=|g|
\label{spin-constraint}
\end{equation}
for all $j=1,\dots,n$. It was shown in \cite[Proposition 6.6]{CS17} that the local coordinates $(\lambda_j,\phi_j/m,[\tilde{w}_i]_{\alpha,j},[\tilde{v}_i]_{j,\alpha})$ on the spin Calogero-Moser space $s\cC_{n}^m$ are canonical, i.e. the reduced symplectic form on $s\cC_{n}^m$ can be locally written as follows
\begin{equation}
\tilde{\omega}=\sum_{j=1}^n\bigg(md\lambda_j\wedge d\phi_j+\sum_{i=0}^{m-1}\sum_{\alpha=1}^d[\tilde{w}_i]_{\alpha,j}\wedge[\tilde{v}_i]_{j,\alpha}\bigg).
\label{}
\end{equation}
The Poisson bracket of two functions $f,g$ on the spin Calogero-Moser space $s\cC_{n}^m$ can be locally computed via
\begin{equation}
\{f,g\}=\sum_{j=1}^n\bigg[m\bigg(\frac{\partial f}{\partial\lambda_j}\frac{\partial g}{\partial\phi_j}
-\frac{\partial f}{\partial\phi_j}\frac{\partial g}{\partial\lambda_j}\bigg)
+\sum_{i=0}^{m-1}\sum_{\alpha=1}^d\bigg(\frac{\partial f}{\partial[\tilde{w}_i]_{\alpha,j}}\frac{\partial g}{\partial[\tilde{v}_i]_{j,\alpha}}
-\frac{\partial f}{\partial[\tilde{v}_i]_{j,\alpha}}\frac{\partial g}{\partial[\tilde{w}_i]_{\alpha,j}}\bigg)\bigg].
\label{sPoisson}
\end{equation}

The expressions of the functions $A,C,D$ on the spin Calogero-Moser space are formally the same as in the spinless case. Namely,
\begin{gather}
A(z)=\det(z\1_V-P),\label{eq:A-s}\\ 
C(z)=\tr(X\,\adj(z\1_V-P)vw),\label{eq:C-s}\\
D(z)=\tr(X\,\adj(z\1_V-P)).\label{eq:D-s}
\end{gather}
Again, the dependence of them on the class of $(X,P,v,w)$ is suppressed. In the spin case the product $vw$ is understood to be the sum of the tensor (or dyadic) products
\begin{equation} vw=\bigoplus_{i=0}^{m-1} v_iw_i,\end{equation}
which can also be seen from the alternative expression
\begin{equation} C(z)=wX\adj(z\1_V-P)v.\end{equation}

The next result gives Theorem \ref{thm:rel} for the spin case.
\begin{proposition}\label{prop:3}
The variables $\phi_1,\dots,\phi_n$ can be expressed using the functions $A$ and $D$ as follows
\begin{equation}
\phi_k=\frac{D(\lambda_k)}{A'(\lambda_k)},\quad k=1,\dots,n.
\label{}
\end{equation}
\end{proposition}

\begin{proof}
Using expression \eqref{eq:dzdef1} we get
\begin{equation}
D(z)=\sum_{i=0}^{m-1}\tr(\tilde{Q}_i A(z)z^{m-2}\blambda(z^m-\blambda^m)^{-1})
=\sum_{i=0}^{m-1}\sum_{j=1}^n(\tilde{Q}_i)_{j,j}\lambda_jz^{m-2}\prod_{\substack{\ell=1\\(\ell\neq j)}}^n(z^m-\lambda_\ell^m).
\label{}
\end{equation}
Interchanging the order of summation and the identity
\begin{equation}
\sum_{i=0}^{m-1}\bigg(c_i+\sum_{r=0}^i[\tilde{v}_r \tilde{w}_r]_{j,j}
-\sum_{s=0}^{m-1}\frac{m-s}{m}[\tilde{v}_s\tilde{w}_s]_{j,j}\bigg)=0,
\label{eq:sumcspin}
\end{equation}
which follows from \eqref{spin-constraint}, give
\begin{equation}
\label{eq:dzdef2}
D(z)=\sum_{j=1}^n\phi_jm\lambda_jz^{m-2}\prod_{\substack{\ell=1\\(\ell\neq j)}}^n(z^m-\lambda_\ell^m).
\end{equation}
Substituting $z=\lambda_k$ into this formula yields
\begin{equation}
D(\lambda_k)=\phi_km\lambda_k^{m-1}\prod_{\substack{\ell=1\\(\ell\neq k)}}^n(\lambda_k^m-\lambda_\ell^m)=\phi_kA'(\lambda_k)
\label{}
\end{equation}
and the proof is complete.
\end{proof}

The next statement completes the proof of Theorem \ref{thm:mod} in the spin case.
\begin{proposition}
For a point $[(X,P,v,w)] \in s\mathcal{C}_n^m$ let 
\begin{equation}
s(z)=\frac{C(z)}{|g|A'(z)} \in \C(z)
\end{equation}
with $A(z)$ and $C(z)$ defined in \eqref{eq:A-s} and \eqref{eq:C-s}, respectively. Let us moreover define the variables $\theta_1,\dots,\theta_n$ as
\begin{equation}
\theta_k=s(\lambda_k),\quad k=1,\dots,n.
\label{Sklyaninspin}
\end{equation}
Then the variables $\theta_1,\dots,\theta_n$ given by the generalized Sklyanin's formula \eqref{Sklyaninspin} are conjugate to $\lambda_1,\dots,\lambda_n$.
\end{proposition}

\begin{proof}
A direct calculation shows that
\begin{equation}
C(z)=\sum_{i=0}^{m-1}\sum_{j,t=1}^n[\tilde{v}_i\tilde{w}_i]_{t,j}(\tilde{Q}_{i-1})_{j,t}\lambda_tz^{m-2}\prod_{\substack{\ell=1\\(\ell\neq t)}}^n(z^m-\lambda_\ell^m)
\label{C(z)}
\end{equation}
hence by using \eqref{eq:azpdef1} $\theta_k$ can be expressed as
\begin{equation}
\theta_k=\frac{1}{m|g|}\sum_{i=0}^{m-1}\sum_{j=1}^n[\tilde{v}_i\tilde{w}_i]_{k,j}(\tilde{Q}_{i-1})_{j,k}.
\label{}
\end{equation}
Taking \eqref{eq:tilde-Q-j,j}--\eqref{spin-constraint} into account, the variable $\theta_k$ can be explicitly spelled out as
\begin{equation}
\theta_k=\frac{\phi_k}{m}+e_k(\lambda_k,\tilde{v},\tilde{w})+f_k(\lambda_1,\dots,\lambda_n,\tilde{v},\tilde{w}),
\label{eq:thetaspindef}
\end{equation}
where
\begin{equation}
\label{eq:ekdef2}
e_k(\lambda_k,\tilde{v},\tilde{w})=\frac{1}{m|g|\lambda_k}\sum_{i=0}^{m-1}[\tilde{v}_i\tilde{w}_i]_{k,k}\bigg(c_{i-1}+\sum_{r=0}^{i-1}[\tilde{v}_r\tilde{w}_r]_{k,k}-\sum_{s=0}^{m-1}\frac{m-s}{m}[\tilde{v}_s\tilde{w}_s]_{k,k}\bigg),
\end{equation}
and
\begin{equation}
\label{eq:fkdef2}
f_k(\lambda_1,\dots,\lambda_n,\tilde{v},\tilde{w})=-\frac{1}{m|g|}\sum_{h,i=0}^{m-1}
\sum_{\substack{t=1\\(t\neq k)}}^n[\tilde{v}_i\tilde{w}_i]_{k,t}[\tilde{v}_{i-h-1}\tilde{w}_{i-h-1}]_{t,k}
\frac{\lambda_t^{m-h-1}\lambda_k^h}{\lambda_t^m-\lambda_k^m}.
\end{equation}

From \eqref{sPoisson} and \eqref{eq:thetaspindef} we get
\begin{equation}
\{\lambda_j,\theta_k\}=m\frac{\partial\theta_k}{\partial\phi_j}=m\frac{1}{m}\frac{\partial\phi_k}{\partial\phi_j}=\delta_{j,k}.
\end{equation}
The explicit expression \eqref{eq:thetaspindef} lets us decompose $\{\theta_j,\theta_k\}$ as follows
\begin{multline}
\{\theta_j,\theta_k\}=\frac{1}{m^2}\{\phi_j,\phi_k\}+\frac{1}{m}\{\phi_j,e_k\}+\frac{1}{m}\{e_j,\phi_k\}+\{e_j,e_k\}+\frac{1}{m}\{\phi_j,f_k\}+\frac{1}{m}\{f_j,\phi_k\}\\+\{e_j,f_k\}+\{f_j,e_k\}+\{f_j,f_k\}.
\end{multline}
Since $\phi_j$ and $\phi_k$ Poisson commute and $e_j$ depends only on $\lambda_j$, but not the other $\lambda$'s, each of the first four terms on the right-hand side is zero, that is
\begin{equation}
\{\phi_j,\phi_k\}=0,\quad \{\phi_j,e_k\}=0,\quad \{e_j,\phi_k\}=0,\quad \{e_j,e_k\}=0.
\end{equation}
Hence we are left with
\begin{equation}
\{\theta_j,\theta_k\}=\frac{1}{m}\left(\{\phi_j,f_k\}+\{f_j,\phi_k\}\right)+\{e_j,f_k\}+\{f_j,e_k\}+\{f_j,f_k\},
\end{equation}
where the terms we grouped cancel, because for every $j,k=1,\dots,n$, we have
\begin{equation}
\{f_j,\phi_k\}+\{\phi_j,f_k\}=\frac{\partial f_k}{\partial \lambda_j}-\frac{\partial f_j}{\partial \lambda_k}=0.
\end{equation}
Indeed, we have
\begin{equation}
\frac{\partial f_j}{\partial \lambda_k}=-\frac{1}{m|g|}\sum_{h,i=0}^{m-1}[\tilde{v}_i\tilde{w}_i]_{j,k}[\tilde{v}_{i-h-1}\tilde{w}_{i-h-1}]_{k,j}\frac{\partial}{\partial\lambda_k}\frac{\lambda_k^{m-h-1}\lambda_j^h}{\lambda_k^m-\lambda_j^m}
\label{f_j-lambda-k}
\end{equation}
and a straightforward calculation shows that
\begin{equation}\label{key}
\frac{\partial}{\partial\lambda_k}\frac{\lambda_k^{m-h-1}\lambda_j^h}{\lambda_k^m-\lambda_j^m}
=
\frac{\partial}{\partial\lambda_j}\frac{\lambda_j^{h+1}\lambda_k^{m-h-2}}{\lambda_j^m-\lambda_k^m}.
\end{equation}
Thus
\begin{equation}
\frac{\partial f_j}{\partial \lambda_k}=-\frac{1}{m|g|}\sum_{h,i=0}^{m-1}[\tilde{v}_i\tilde{w}_i]_{j,k}[\tilde{v}_{i-h-1}\tilde{w}_{i-h-1}]_{k,j}\frac{\partial}{\partial\lambda_j}\frac{\lambda_j^{h+1}\lambda_k^{m-h-2}}{\lambda_j^m-\lambda_k^m}.
\end{equation}
Rewriting the sum using a new pair of indices $h',i'$ given by
\begin{equation}
\begin{aligned}
h'&\equiv m-h-2\pmod{m}\\
i'&\equiv i-h-1\pmod{m}
\end{aligned}
\end{equation}
we get
\begin{equation}\label{key}
\frac{\partial f_j}{\partial \lambda_k}=-\frac{1}{m|g|}\sum_{h',i'=0}^{m-1}[\tilde{v}_{i'}\tilde{w}_{i'}]_{k,j}[\tilde{v}_{i'-h'-1}\tilde{w}_{i'-h'-1}]_{j,k}\frac{\partial}{\partial\lambda_j}\frac{\lambda_j^{m-h'-1}\lambda_k^{h'}}{\lambda_j^m-\lambda_k^m},
\end{equation}
which, by an exchange of $j$ and $k$ in \eqref{f_j-lambda-k}, can be seen to coincide with $\partial f_k/\partial\lambda_j$. As a consequence, we now have
\begin{equation}
\{\theta_j,\theta_k\}=\{e_j,f_k\}+\{f_j,e_k\}+\{f_j,f_k\}.
\label{th_j-th_k}
\end{equation}

Let us consider the first term on the right-hand side. Since $e_j$ and $f_k$ do not depend on any of the $\phi$'s and $e_j$ only depends on the $j$-th column (resp. row) of $[\tilde{w}_i]$ (resp. $[\tilde{v}_i]$) we have
\begin{equation}
\{e_j,f_k\}=\sum_{i=0}^{m-1}\sum_{\alpha=1}^d\left(\frac{\partial e_j}{\partial[\tilde{w}_i]_{\alpha,j}}\frac{\partial f_k}{\partial[\tilde{v}_i]_{j,\alpha}}-\frac{\partial e_j}{\partial[\tilde{v}_i]_{j,\alpha}}\frac{\partial f_k}{\partial[\tilde{w}_i]_{\alpha,j} }\right).
\label{e_j-f_k}
\end{equation}

A straightforward computation yields
\begin{multline}
\frac{\partial e_j}{\partial[\tilde{w}_i]_{\alpha,j}}=\frac{1}{m|g|\lambda_j}\bigg[[\tilde{v}_i]_{j,\alpha}\bigg(c_{i-1}+\sum_{r=0}^{i-1}[\tilde{v}_r\tilde{w}_r]_{j,j}-\sum_{s=0}^{m-1}\frac{m-s}{m}[\tilde{v}_s\tilde{w}_s]_{j,j}\bigg)\\
-\frac{m-i}{m}[\tilde{v}_i]_{j,\alpha}\sum_{h=0}^{m-1}[\tilde{v}_h\tilde{w}_h]_{j,j}+[\tilde{v}_i]_{j,\alpha}\sum_{h=i+1}^{m-1}[\tilde{v}_i\tilde{w}_i]_{j,j}\bigg].
\end{multline}
Collecting the common factor $[\tilde{v}_i]_{j,\alpha}$ and applying \eqref{spin-constraint} give us
\begin{equation}
\frac{\partial e_j}{\partial[\tilde{w}_i]_{\alpha,j}}=\frac{[\tilde{v}_i]_{j,\alpha}}{m|g|\lambda_j}\bigg[c_{i-1}+\frac{i}{m}|g|-[\tilde{v}_i\tilde{w}_i]_{j,j}-\sum_{s=0}^{m-1}\frac{m-s}{m}[\tilde{v}_s\tilde{w}_s]_{j,j}\bigg].
\label{e_j-w}
\end{equation}
Similarly,
\begin{equation}
\frac{\partial e_j}{\partial [\tilde{v}_i]_{j,\alpha}}=\frac{[\tilde{w}_i]_{\alpha,j}}{m|g|\lambda_j}\bigg[c_{i-1}+\frac{i}{m}|g|-[\tilde{v}_i\tilde{w}_i]_{j,j}-\sum_{s=0}^{m-1}\frac{m-s}{m}[\tilde{v}_s\tilde{w}_s]_{j,j}\bigg].
\label{e_j-v}
\end{equation}
As for the partial derivatives of $f_k$, we have
\begin{equation}
\frac{\partial f_k}{\partial[\tilde{w}_i]_{\alpha,j}}=-\frac{[\tilde{v}_i]_{k,\alpha}}{m|g|}\sum_{h=0}^{m-1}[\tilde{v}_{i-h-1}\tilde{w}_{i-h-1}]_{j,k}\frac{\lambda_j^{m-h-1}\lambda_k^h}{\lambda_j^m-\lambda_k^m}
\label{f_k-w}
\end{equation}
and
\begin{equation}
\frac{\partial f_k}{\partial[\tilde{v}_i]_{j,\alpha}}=-\frac{[\tilde{w}_i]_{\alpha,k}}{m|g|}\sum_{h=0}^{m-1}[\tilde{v}_{i+h+1}\tilde{w}_{i+h+1}]_{k,j}\frac{\lambda_j^{m-h-1}\lambda_k^h}{\lambda_j^m-\lambda_k^m}.
\label{f_k-v}
\end{equation}
Putting formulas \eqref{e_j-w}--\eqref{f_k-v} together, $\{e_j,f_k\}$ \eqref{e_j-f_k} is found to be
\begin{multline}
\{e_j,f_k\}=\frac{1}{(m|g|)^2}\sum_{h,i=0}^{m-1}\big([\tilde{v}_i\tilde{w}_i]_{k,j}[\tilde{v}_{i-h-1}\tilde{w}_{i-h-1}]_{j,k}-[\tilde{v}_i\tilde{w}_i]_{j,k}[\tilde{v}_{i+h+1}\tilde{w}_{i+h+1}]_{k,j}\big)\times\\
\times\bigg(c_{i-1}+\frac{i}{m}|g|-[\tilde{v}_i\tilde{w}_i]_{j,j}-\sum_{s=0}^{m-1}\frac{m-s}{m}[\tilde{v}_s\tilde{w}_s]_{j,j}\bigg)\frac{\lambda_j^{m-h-2}\lambda_k^h}{\lambda_j^m-\lambda_k^m}.
\label{e_j-f_k-1}
\end{multline}
The Poisson bracket $\{f_j,e_k\}$ is obtained from $\{e_j,f_k\}$ \eqref{e_j-f_k-1} by changing its sign and exchanging $j$ and $k$. Hence we get
\begin{multline}
\{f_j,e_k\}=-\frac{1}{(m|g|)^2}\sum_{h',i=0}^{m-1}\big([\tilde{v}_i\tilde{w}_i]_{k,j}[\tilde{v}_{i+h'+1}\tilde{w}_{i+h'+1}]_{j,k}-[\tilde{v}_i\tilde{w}_i]_{j,k}[\tilde{v}_{i-h'+1}\tilde{w}_{i-h'+1}]_{k,j}\big)\times\\
\times\bigg(c_{i-1}+\frac{i}{m}|g|-[\tilde{v}_i\tilde{w}_i]_{k,k}-\sum_{s=0}^{m-1}\frac{m-s}{m}[\tilde{v}_s\tilde{w}_s]_{k,k}\bigg)\frac{\lambda_k^{m-h'-2}\lambda_j^{h'}}{\lambda_j^m-\lambda_k^m}.
\label{f_j-e_k-1}
\end{multline}
Rewriting this using the new index $h\equiv m-h'-2\pmod{m}$ allows us to collect the factors of the terms with the same $\lambda$ dependence in $\{e_j,f_k\}$ and $\{f_j,e_k\}$. Then we add \eqref{e_j-f_k-1} and \eqref{f_j-e_k-1} together and find that the terms with $c_i$ and $\frac{i}{m}|g|$ cancel and as a result, we get
\begin{multline}
\{e_j,f_k\}+\{f_j,e_k\}=\frac{1}{(m|g|)^2\lambda_j\lambda_k}\sum_{h=1}^{m-1}\sum_{i=0}^{m-1}
([\tilde{v}_i\tilde{w}_i]_{k,k}-[\tilde{v}_i\tilde{w}_i]_{j,j})\times\\
\times([\tilde{v}_i\tilde{w}_i]_{k,j}[\tilde{v}_{i-h}\tilde{w}_{i-h}]_{j,k}-[\tilde{v}_i\tilde{w}_i]_{j,k}[\tilde{v}_{i+h}\tilde{w}_{i+h}]_{k,j})\frac{\lambda_j^{m-h}\lambda_k^h}{\lambda_j^m-\lambda_k^m}.
\label{ejfk+fjek}
\end{multline}

Let us now consider the last term $\{f_j,f_k\}$ in $\{\theta_j,\theta_k\}$ \eqref{th_j-th_k}. Since $f_j$ and $f_k$ do not depend on any of the $\phi$'s we have
\begin{equation}
\{f_j,f_k\}=\sum_{\ell=1}^n\sum_{i=0}^{m-1}\sum_{\alpha=1}^d\left(\frac{\partial f_j}{\partial[\tilde{w}_i]_{\alpha,\ell}}\frac{\partial f_k}{\partial[\tilde{v}_i]_{\ell,\alpha}}-\frac{\partial f_j}{\partial[\tilde{v}_i]_{\ell,\alpha}}\frac{\partial f_k}{\partial[\tilde{w}_i]_{\alpha,\ell}}\right).
\label{f_j-f_k}
\end{equation}
We already calculated most of these partial derivatives in \eqref{f_k-w} and \eqref{f_k-v}. The only ones remaining are
\begin{equation}
\frac{\partial f_j}{\partial[\tilde{w}_i]_{\alpha,j}}=-\frac{1}{m|g|}\sum_{\substack{t=1\\(t\neq j)}}^n[\tilde{v}_i]_{t,\alpha}\sum_{h=0}^{m-1}[\tilde{v}_{i+h+1}\tilde{w}_{i+h+1}]_{j,t}\frac{\lambda_t^{m-h-1}\lambda_j^h}{\lambda_t^m-\lambda_j^m}
\end{equation}
and
\begin{equation}
\frac{\partial f_j}{\partial[\tilde{v}_i]_{j,\alpha}}=-\frac{1}{m|g|}\sum_{\substack{t=1\\(t\neq j)}}^n[\tilde{w}_i]_{\alpha,t}\sum_{h=0}^{m-1}[\tilde{v}_{i-h-1}\tilde{w}_{i-h-1}]_{t,j}\frac{\lambda_t^{m-h-1}\lambda_j^h}{\lambda_t^m-\lambda_j^m}.
\end{equation}
for $\alpha=1,\dots,d$. Now we break up the sum \eqref{f_j-f_k} into six parts, namely
\begin{multline}
\{f_j,f_k\}=\underbrace{\sum_{\substack{\ell=1\\(\ell\neq j,k)}}^n\sum_{i=0}^{m-1}\sum_{\alpha=1}^d\frac{\partial f_j}{\partial[\tilde{w}_i]_{\alpha,\ell}}\frac{\partial f_k}{\partial[\tilde{v}_i]_{\ell,\alpha}}}_{=:A}
-\underbrace{\sum_{\substack{\ell=1\\(\ell\neq j,k)}}^n\sum_{i=0}^{m-1}\sum_{\alpha=1}^d\frac{\partial f_j}{\partial[\tilde{v}_i]_{\ell,\alpha}}\frac{\partial f_k}{\partial[\tilde{w}_i]_{\alpha,\ell}}}_{=:B}\\
+\underbrace{\sum_{i=0}^{m-1}\sum_{\alpha=1}^d\frac{\partial f_j}{\partial[\tilde{w}_i]_{\alpha,j}}\frac{\partial f_k}{\partial[\tilde{v}_i]_{j,\alpha}}}_{=:C}
-\underbrace{\sum_{i=0}^{m-1}\sum_{\alpha=1}^d\frac{\partial f_j}{\partial[\tilde{v}_i]_{j,\alpha}}\frac{\partial f_k}{\partial[\tilde{w}_i]_{\alpha,j}}}_{=:D}\\
+\underbrace{\sum_{i=0}^{m-1}\sum_{\alpha=1}^d\frac{\partial f_j}{\partial[\tilde{w}_i]_{\alpha,k}}\frac{\partial f_k}{\partial[\tilde{v}_i]_{k,\alpha}}}_{=:E}
-\underbrace{\sum_{i=0}^{m-1}\sum_{\alpha=1}^d\frac{\partial f_j}{\partial[\tilde{v}_i]_{k,\alpha}}\frac{\partial f_k}{\partial[\tilde{w}_i]_{\alpha,k}}}_{=:F}.
\label{eq:term2}
\end{multline}
Fortunately, these expressions are related. For example, we get $B$ if we exchange $j$ and $k$ in $A$. We denote this by writing that $B=(A)_{j\leftrightarrow k}$. There are similar relations between the expressions $C$ and $F$, as well as between the expressions $D$ and $E$. In short, we have
\begin{equation}
B=(A)_{j\leftrightarrow k},\quad
C=(F)_{j\leftrightarrow k},\quad
D=(E)_{j\leftrightarrow k}.
\label{relations}
\end{equation}
This observation saves us half the work as we only need to calculate, say $A$, $F$, and $E$.
First, we calculate $A$ and find that
\begin{equation}
A=\frac{1}{(m|g|)^2}\sum_{\substack{\ell=1\\(\ell\neq j,k)}}^n\sum_{i,h,h'=0}^{m-1}[\tilde{v}_i\tilde{w}_i]_{j,k}[\tilde{v}_{i-h-1}\tilde{w}_{i-h-1}]_{\ell,j}[\tilde{v}_{i+h'+1}\tilde{w}_{i+h'+1}]_{k,\ell}\frac{\lambda_\ell^{2m-h-h'-2}\lambda_j^h\lambda_k^{h'}}{(\lambda_\ell^m-\lambda_j^m)(\lambda_\ell^m-\lambda_k^m)}.
\label{A}
\end{equation}
Second, we calculate $F$ and get
\begin{equation}
F=\frac{1}{(m|g|)^2}\sum_{\substack{\ell=1\\(\ell\neq k)}}^n\sum_{i,h,h'=0}^{m-1}[\tilde{v}_i\tilde{w}_i]_{\ell,j}[\tilde{v}_{i+h+1}\tilde{w}_{i+h+1}]_{k,\ell}[\tilde{v}_{i+h'+1}\tilde{w}_{i+h'+1}]_{j,k}\frac{\lambda_\ell^{m-h-1}\lambda_k^{m-h'-1+h}\lambda_j^{h'}}{(\lambda_\ell^m-\lambda_k^m)(\lambda_k^m-\lambda_j^m)}.
\label{F}
\end{equation}
Third, we calculate $E$ and get
\begin{equation}
E=\frac{1}{(m|g|)^2}\sum_{\substack{\ell=1\\(\ell\neq k)}}^n\sum_{i,h,h'=0}^{m-1}[\tilde{v}_i\tilde{w}_i]_{j,\ell}[\tilde{v}_{i-h-1}\tilde{w}_{i-h-1}]_{\ell,k}[\tilde{v}_{i-h'-1}\tilde{w}_{i-h'-1}]_{k,j}\frac{\lambda_\ell^{m-h-1}\lambda_k^{m-h'-1+h}\lambda_j^{h'}}{(\lambda_\ell^m-\lambda_k^m)(\lambda_k^m-\lambda_j^m)}.
\label{E}
\end{equation}
We obtain explicit formulas for $B$, $C$, and $D$ from \eqref{A}--\eqref{E} and the relations \eqref{relations}. Namely,
\begin{equation}
B=\frac{1}{(m|g|)^2}\sum_{\substack{\ell=1\\(\ell\neq j,k)}}^n\sum_{i,h,h'=0}^{m-1}[\tilde{v}_i\tilde{w}_i]_{k,j}[\tilde{v}_{i-h'-1}\tilde{w}_{i-h'-1}]_{\ell,k}[\tilde{v}_{i+h+1}\tilde{w}_{i+h+1}]_{j,\ell}\frac{\lambda_\ell^{2m-h-h'-2}\lambda_j^h\lambda_k^{h'}}{(\lambda_\ell^m-\lambda_j^m)(\lambda_\ell^m-\lambda_k^m)},
\label{B}
\end{equation}
\begin{equation}
C=\frac{1}{(m|g|)^2}\sum_{\substack{\ell=1\\(\ell\neq j)}}^n\sum_{i,h,h'=0}^{m-1}[\tilde{v}_i\tilde{w}_i]_{\ell,k}[\tilde{v}_{i+h+1}\tilde{w}_{i+h+1}]_{j,\ell}[\tilde{v}_{i+h'+1}\tilde{w}_{i+h'+1}]_{k,j}\frac{\lambda_\ell^{m-h-1}\lambda_j^{m-h'-1+h}\lambda_k^{h'}}{(\lambda_\ell^m-\lambda_j^m)(\lambda_j^m-\lambda_k^m)},
\label{C}
\end{equation}
and
\begin{equation}
D=\frac{1}{(m|g|)^2}\sum_{\substack{\ell=1\\(\ell\neq j)}}^n\sum_{i,h,h'=0}^{m-1}[\tilde{v}_i\tilde{w}_i]_{k,\ell}[\tilde{v}_{i-h-1}\tilde{w}_{i-h-1}]_{\ell,j}[\tilde{v}_{i-h'-1}\tilde{w}_{i-h'-1}]_{j,k}\frac{\lambda_\ell^{m-h-1}\lambda_j^{m-h'-1+h}\lambda_k^{h'}}{(\lambda_\ell^m-\lambda_j^m)(\lambda_j^m-\lambda_k^m)}.
\label{D}
\end{equation}

By a suitable change of indices in $D$ and $F$ we see that in $A-D-F$ almost all terms cancel. The only ones remaining are the terms with $\ell=k$ in $D$ and the terms with $\ell=j$ in $F$. As a consequence, we get
\begin{multline}
A-D-F=\frac{1}{(m|g|)^2}\sum_{i,h,h'=0}^{m-1}\big([\tilde{v}_i\tilde{w}_i]_{k,k}[\tilde{v}_{i-h-1}\tilde{w}_{i-h-1}]_{k,j}[\tilde{v}_{i-h'-1}\tilde{w}_{i-h'-1}]_{j,k}+\\+[\tilde{v}_i\tilde{w}_i]_{j,j}[\tilde{v}_{i+h'+1}\tilde{w}_{i+h'+1}]_{k,j}[\tilde{v}_{i+h+1}\tilde{w}_{i+h+1}]_{j,k}\big)\frac{\lambda_j^{m-h'-1+h}\lambda_k^{m-h-1+h'}}{(\lambda_j^m-\lambda_k^m)^2}.
\label{A-D-F}
\end{multline}
With the same type of computation we obtain
\begin{multline}
-B+C+E=-\frac{1}{(m|g|)^2}\sum_{i,h,h'=0}^{m-1}\big([\tilde{v}_i\tilde{w}_i]_{k,k}[\tilde{v}_{i+h'+1}\tilde{w}_{i+h'+1}]_{k,j}[\tilde{v}_{i+h+1}\tilde{w}_{i+h+1}]_{j,k}+\\+[\tilde{v}_i\tilde{w}_i]_{j,j}[\tilde{v}_{i-h-1}\tilde{w}_{i-h-1}]_{k,j}[\tilde{v}_{i-h'-1}\tilde{w}_{i-h'-1}]_{j,k}\big)\frac{\lambda_j^{m-h'-1+h}\lambda_k^{m-h-1+h'}}{(\lambda_j^m-\lambda_k^m)^2}.
\label{-B+C+E}
\end{multline}
Since the exponents of $\lambda_j$ and $\lambda_k$ do not depend on $i$ and depend only on the difference of $h$ and $h'$, but not on the individual indices, introducing a new index $h'':=h-h'$ and adding \eqref{A-D-F} to \eqref{-B+C+E}, yields an explicit formula for the Poisson bracket $\{f_j,f_k\}$. Namely, we get
\begin{multline}
\{f_j,f_k\}=-\frac{1}{(m|g|)^2\lambda_j\lambda_k}\sum_{h''=1}^{m-1}\sum_{i=0}^{m-1}
([\tilde{v}_i\tilde{w}_i]_{k,k}-[\tilde{v}_i\tilde{w}_i]_{j,j})\times\\
\times([\tilde{v}_i\tilde{w}_i]_{k,j}[\tilde{v}_{i-h''}\tilde{w}_{i-h''}]_{j,k}-[\tilde{v}_i\tilde{w}_i]_{j,k}[\tilde{v}_{i+h''}\tilde{w}_{i+h''}]_{k,j})\frac{\lambda_j^{m-h''}\lambda_k^{h''}}{\lambda_j^m-\lambda_k^m}.
\label{eq:fjfkfinal}
\end{multline}

This is the same expression as \eqref{ejfk+fjek} only with opposite sign. Hence these two terms in $\{\theta_j,\theta_k\}$ cancel and we obtain
\begin{equation}
\{\theta_j,\theta_k\}=0.
\end{equation}
Finally, let us observe that due to \eqref{spin-constraint} we can take any fixed $h\in\Z_m$ and $\beta\in\{1,\dots,d\}$ and express $[\tilde{v}_h]_{j,\beta}$ in terms of $[\tilde{v}_i]_{j,\alpha}$ and $[\tilde{w}_{i'}]_{\alpha',j}$ with $i,i'\in\Z_m$ ($i\neq h$) and $\alpha,\alpha'\in\{1,\dots,d\}$ ($\alpha\neq\beta$) for all $j=1,\dots,n$. This means that $[\tilde{v}_h]_{j,\beta}$ ($j=1,\dots,n$) are not independent coordinates on $s\cC_m^n$, i.e.
\begin{equation}
\{\theta_k,[\tilde{v}_h]_{j,\beta}\}=0,\quad \{\theta_k,[\tilde{w}_h]_{\beta,j}\}=0
\end{equation}
for all $j,k=1,\dots,n$.
\end{proof}

\begin{remark}
Let us list some important special cases of our results. In \cite{CS17}, it was shown that the $m=2$, $d=1$ case corresponds to the rational Calogero-Moser system of type $B_n$ (and with $g_1=0$ of type $D_n$). Setting $m=1$, $d>1$ produces the Gibbons-Hermsen system \cite{GH84}, whereas the $m=2$, $d>1$ case contains the type $B_n$ variant of the Gibbons-Hermsen system.
\end{remark}

\section{The equivariant geometry of the interpolation curves}
\label{sec:interpolation}

Now we briefly describe the geometry of the interpolation curves appearing in Theorems \ref{thm:rel} and \ref{thm:mod}. These are the affine plane curves
\begin{equation}
C_1=\{ (z,r(z))\colon z \in \C \} \subset \C^2
\quad\text{and}\quad
C_2=\{ (z,s(z))\colon z \in \C\} \subset \C^2.
\end{equation}
Both of these are rationally parametrized. Hence, they can be completed to rational curves in $\C\P^2$. The expressions \eqref{eq:azpdef1}, \eqref{C(z)-1}, \eqref{C(z)}, \eqref{D(z)} and \eqref{eq:dzdef2} show that the polynomials $A'(z)$, $C(z)$ and $D(z)$ in all cases are divisible by $z^{m-2}$. After cancellations we can write
\begin{equation}
r(z)=:\frac{p_1(z^m)}{z q(z^m)}\quad\text{and}\quad s(z)=:\frac{p_2(z^m)}{z q(z^m)},
\end{equation}
where $p_1(z)$, $p_2(z)$ and $q(z)$ are polynomials of degree $n-1$.
The defining equation of the curve $C_\delta$ is
\begin{equation}
q(z^m)zy-p_\delta(z^m)=0,\quad \delta=1,2.
\end{equation}

Let $\Delta$ be the root system $A_{m-1}$ and let us choose a primitive $m$-th root of unity $\omega$. 
There corresponds to $\Delta$ a subgroup $G_\Delta$ of $\SL(2,\C)$, a cyclic subgroup of order $m$, 
which is generated by the matrix
\begin{equation}
\sigma=
\begin{pmatrix}
\omega&0\\0&\omega^{-1}
\end{pmatrix}.
\end{equation}
All irreducible representations of $G_\Delta$ are one-dimensional, and are given by $\rho_j\colon \sigma\mapsto \omega^j$, for $j\in\Z_m$. The corresponding McKay quiver is the cyclic Dynkin diagram of type $\widetilde A_{m-1}^{(1)}$. The group $G_\Delta$ acts on $\C^2$; the quotient variety $\C^2/G_\Delta$ has an isolated singularity of type $A_{m-1}$ at the origin. In coordinates, the ring of functions $H^0(\mathcal{O}_{\C^2/G_\Delta})=\mathbb{C}[y,z]^{G_\Delta}$ is generated by $a=z^m$, $b=y^m$ and $c=zy$ which satisfy the relation
\begin{equation}
\label{eq:amsing}
ab=c^m.
\end{equation}

As it was remarked in \cite[Section 5.1]{CS17} the set of eigenvalues $(\lambda_1^m,\dots,\lambda_n^m)$ of the transformation $P_{0}P_1\dots P_{m-1} \in \mathrm{End}(V_0)$  determines $(\lambda_1,\dots,\lambda_n)$ only up to permutations and multiplication by $m$-th root of unity. Therefore, the coordinates $\lambda_i, \phi_i$ are only well-defined up to the action of $S_n \ltimes G_\Delta$, where the $S_n$ component permutes $\{\lambda_i\}$ and  $\{\phi_i\}$ simultaneously, and the generator of the $G_\Delta$ component maps $(\lambda_i, \phi_i)$ to $(\omega \lambda_i, \omega^{-1}\phi_i)$. 

The following lemma is straightforward from \eqref{15} and \eqref{eq:thetaspindef}.
\begin{lemma}
When $\lambda_i$ is replaced by $\omega \lambda_i$, then $\theta_i$ is replaced by $\omega^{-1} \theta_i$. Therefore, the coordinates $\lambda_i$, $\theta_i$ are also  well-defined only up to the action of $S_n \ltimes G_\Delta$.
\end{lemma}
\begin{corollary} 
The pairs of variables $(\lambda_i,\phi_i)$ and $(\lambda_i,\theta_i)$ are well-defined on $\C^2/G_\Delta$. 
\end{corollary}

As a result, the curves $C_1$ and $C_2$ are only well-defined up to the action of $G_\Delta$. But they descend to well-defined curves on the quotient space $\C^2/G_\Delta$.
\begin{corollary}
The curves $C_1$ and $C_2$ descend to well-defined rational curves $C_1/G_\Delta$ and $C_2/G_\Delta$ on $\C^2/G_\Delta$. When considered as a subvariety of $\mathbb{C}^3=\mathrm{Spec}(\mathbb{C}[a,b,c])$, $C_1/G_\Delta$ and $C_2/G_\Delta$ are given by the intersection of the surface \eqref{eq:amsing} and the surface
\begin{equation}
q(a)c-p_\delta(a)=0,\quad \delta=1,2,
\end{equation}
respectively, or equivalently, the surface swept out by the translations of the graph of the degree $n-1$ interpolating function
\begin{equation}
c=\frac{p_\delta(a)}{q(a)},\quad \delta=1,2
\end{equation}
in the $b$-direction. In this way we obtain a map
\begin{equation}
\mathcal{C}^m_n \to \mathrm{RatCurves}^{n}(\C^2/G_\Delta)
\end{equation}
defined on the dense open subset $\mathcal{C}_n^{m,P} \subset \mathcal{C}_n^m$, where $\mathrm{RatCurves}^{n}(\C^2/G_\Delta)$ is the space of rational curves of degree $n$ on $\C^2/G_\Delta$.

Conversely, if $C \subset \C^2/G_\Delta$ is a rational curve of degree $n$ which is of the above form, then any distinct $n$ points on it determine a point of $\mathcal{C}_n^{m,P}$, such that the associated curve $C_\delta/G_\Delta$ (resp. $C_2/G_\Delta$) to this point is $C$. This correspondence associates the point of $\cC_n^m$ with coordinates $\{(\lambda_i, \phi_i)\}$ (resp. $\{(\lambda_i, \theta_i)\}$) to the $n$ points $\{(\lambda_i, \phi_i)\} \subset C$ (resp. $\{(\lambda_i, \theta_i)\} \subset C$).
\end{corollary}

\end{document}